\documentclass[]{article}
\pdfoutput=1
\usepackage[utf8]{inputenc}
\usepackage[greek, english]{babel}
\usepackage{alphabeta}
\usepackage[margin=1in]{geometry}
\usepackage{csquotes}
\usepackage[T1]{fontenc}
\usepackage{physics}
\usepackage{xcolor}
\usepackage{mathtools}
\usepackage{parskip}
\usepackage{algorithm}
\usepackage{algpseudocodex}
\usepackage{graphicx}
\usepackage{multirow}
\usepackage{subcaption}
\usepackage[labelfont=bf,format=plain,justification=raggedright,singlelinecheck=false]{caption}
\usepackage{tabularx}
\usepackage{booktabs}
\usepackage{float}
\usepackage{amsmath, amssymb, amsthm, bbm, braket, complexity}
\usepackage[numbers,compress]{natbib}
\usepackage{hyperref, cleveref}

\DeclarePairedDelimiter\pars{\lparen}{\rparen}
\newcommand{\bigotilde}[1]{\tilde{\mathcal{O}}\pars*{#1}}
\newcommand{\bigo}[1]{\mathcal{O}\pars*{#1}}

\usepackage{thmtools, thm-restate}

\newtheorem{prob}{Problem}
\crefname{prob}{problem}{problems}
\newtheorem{theorem}{Theorem}
\newtheorem{lemma}{Lemma}
\newtheorem{conjecture}{Conjecture}

\newfloat{algorithm}{t}{lop}

\title{
    Parameterized quantum algorithms for closest string problems
}

\author{
Josh Cudby\thanks{Department of Applied Mathematics and Theoretical Physics, University of Cambridge, UK ({jjcc2@cam.ac.uk}).}
\and Sergii Strelchuk\thanks{Department of Computer Science, University of Oxford, UK.}
}
\date{\today}

\begin{document}

\maketitle

\begin{abstract}
Parameterized complexity enables the practical solution of generally intractable $\NP$-hard problems when certain parameters are small, making it particularly useful in real-world applications.
The study of string problems in this framework has been particularly fruitful, yielding many state-of-the-art classical algorithms that run efficiently in certain parameter regimes contrary to their worst- or average-case performance.
Motivated by the dramatic increase in genomic data and its growing computational demands, we initiate the study of the \emph{quantum} parameterized complexity of the Closest String Problem (CSP) and the related Closest Substring Problem (CSSP). 
We present three quantum algorithms for the CSP and one for the CSSP. 
Each algorithm demonstrates improved performance over classical counterparts in specific parameter regimes, highlighting the promise of quantum approaches in structured combinatorial settings.
We also derive a conditional lower bound for the CSP with binary alphabets, showing that our first algorithm is tight in its dominant scaling factor.
\end{abstract}



\section{Introduction}
\label{sec:intro}
String processing is one of the oldest and most fundamental fields within theoretical computer science, but remains an active area of research.
Recent developments have largely been motivated by rapidly growing quantities of data that are available.
In particular, modern applications in bioinformatics have driven a surge in interest, with the volume of DNA data available growing extremely rapidly~\cite{sayers_genbank_2023}. Indeed, since the dawn of modern sequencing technology, the cost of acquiring new sequences has dropped faster than Moore’s Law for almost two decades~\cite{wetterstrand_dna_nodate}.
The falling cost, combined with recent global pandemics, has prompted massive amounts of viral sequencing and led to a huge data backlog.
The ability to gain meaningful insights from data requires faster algorithms.

Specifically, given a set of strings (often representing DNA, mRNA, or protein sequences), we aim to find a new string that is ``close'' to all inputs, or to a substring of each input.
These tasks correspond to the Closest String Problem (CSP) and the Closest Substring Problem (CSSP)~\cite{kevin_lanctot_distinguishing_2003}, where closeness is measured by distance metrics such as the Hamming or Levenshtein distances~\cite{levenshtein_binary_1965}.
The solution can play the role of a \emph{consensus} string for the inputs, acting as a proxy for a full multiple sequence alignment (MSA).
MSA is the process of aligning 3 or more sequences to identify regions of similarity that generally correspond to functional or structural similarity.
A complete MSA procedure is often prohibitively computationally expensive, requiring $\Omega(n^k)$ time to align $k$ strings of length $n$.
A consensus string can be cheaper to compute while still being useful as a representation of an ``average case'' organism, such as in the study of DNA binding sites~\cite{day_critical_1992}.
Indeed, one MSA heuristic is the ``center star'' method.
In this algorithm, each input sequence is iteratively aligned against a putative center string, and the final alignment is then computed from the results~\cite{zou_novel_2012}.
The quality of the final alignment is closely tied to that of the center string used; using the solution to the CSP guarantees a good alignment.

Both problems are $\NP$-hard -- even over binary alphabets under Hamming distance -- yet the lens of parameterized complexity has enabled more nuanced algorithmic approaches~\cite{gramm_fixed-parameter_2003, nicolas_hardness_2005, lokshtanov_slightly_2018, maji_listing_2015, cheng_efficient_2004, wang_efficient_2009}.  
We aim to obtain a fine-grained understanding of the hardness of problems, and in particular to identify certain structural parameters, such as the number of input sequences or a threshold for the allowable distance of a solution, that can lead to more efficient algorithms.
The choice of metric can significantly impact the complexity of the task. 
For example, the CSP with a fixed number of input strings is known to be solvable in time polynomial in the length of the strings if the distance is measured by the Hamming distance~\cite{gramm_fixed-parameter_2003}. 
Conversely, it is known that no such polynomial-time algorithm can exist when distance is measured by Levenshtein distance~\cite{nicolas_hardness_2005} unless $\P = \NP$.
A careful study of these problems, across different parameter regimes and various metrics, is therefore warranted.

In practical tasks that arise in genomic sequencing and alignment, the inputs processed by string‐matching and alignment routines almost never exhibit the pathological structure required to achieve worst‐case runtimes. 
The average‐case performance, though often more indicative of typical behaviour, is critically dependent on our ability to accurately characterise structured biological data distributions. This is a highly non-trivial task for the diverse and complex probabilistic ensembles representing genomic sequences, which often exhibit specific biases and correlations. Genomic data, for example, carries highly constrained patterns of conservation and repetition arising from evolutionary pressures, and even noisy reads tend to cluster around a small number of high-quality reference loci. As a result, neither measure alone sufficiently represents the practical utility of an algorithm, especially when specific structural parameters inherent to the data (e.g. degree of sequence similarity, the density of conserved regions, or the frequency of repeating k-mers) are present. On the other hand, a parameterized framework that explicitly quantifies these instance-specific measures therefore provides a far more accurate predictor of real-world performance, leading to better heuristics for the corresponding algorithms.

For almost all organisms, genetic material comes in the form of DNA which is composed of four bases \{A, C, G, T\}. Through transcription to mRNA and translation to proteins, these sequences govern cellular behaviour and evolution.
Thus, DNA is the fundamental language of life and understanding how DNA structures vary within and between species plays a central role in modern approaches to human health and medicine.

A key bioinformatics challenge is to identify sequences that are conserved across related organisms. This problem has a variety of important applications, ranging from drug design, species comparison to disease tracking.
In this paper, we study the optimisation problems that correspond to this task, and related tasks such as motif discovery, haplotype inference and primer design in polymerase chain reactions.

Recently, rapid progress in both the theory and hardware of quantum computers has led to the development of quantum algorithms for string processing.
Notable early works include sequence comparison~\cite{hollenberg_fast_2000} and exact pattern matching~\cite{ramesh_string_2003}, whereas recent works have targeted approximate pattern matching~\cite{kociumaka_communication_2024}, seed-and-extend alignment~\cite{prousalis_novel_2024} and \emph{de novo} assembly~\cite{fang_divide-and-conquer_2024}.
A wealth of other results are available in this active field~\cite{ montanaro_quantum_2015,aaronson_quantum_2019, niroula_quantum_2021, sarkar_qibam_2021, childs_quantum_2022, darbari_quantum_2022, jin_quantum_2022, akmal_near-optimal_2023, gall_quantum_2023, gibney_near-optimal_2023}.

In this work, we develop quantum algorithms for the CSP and CSSP.
We discuss their parameterized complexity with respect to alphabet size, number of input strings and threshold distance.
Our contributions are formalised in the following propositions, which provide quantum algorithms for the CSP and the CSSP under various metrics and parameter settings.
In each case, the alphabet size is $\sigma$, the number of input strings is $k$, and the threshold of allowed distance for solutions is $d$.
\begin{restatable}[CSP, any distance]{prop}{csptrivial}
    \label{prop:csp_trivial}
    There exists a quantum algorithm for solving the CSP under any desired distance metric, with runtime $\bigotilde{\sigma^{n/2}\cdot\sqrt{k}\cdot D(n)}$, where $D(n)$ is the \emph{classical} time required to compute the distance between 2 strings of length $n$. In particular, $D(n) = \bigo{n}$ for the Hamming distance and $\bigo{n^2}$ for the Levenshtein distance.
\end{restatable}

\begin{restatable}[CSP, Hamming distance]{prop}{csphamming}
\label{prop:csp_hamming}
    There exists a quantum algorithm for solving the CSP under the Hamming distance, with runtime $\bigotilde{kn + k^2 \cdot d^{(d+5)/2} }$.
\end{restatable}

\begin{restatable}[CSP, Levenshtein distance]{prop}{csplevenshtein}
\label{prop:csp_levenshtein}
    There exists a quantum algorithm for solving the CSP under the Levenshtein distance, with runtime $\bigotilde{\sigma \cdot 2^k \cdot T(n)^{2k}}$, where $T(n)$ satisfies $n/e \leq T(n) < n$ for every $n$.
\end{restatable}

\begin{restatable}[CSSP, Hamming distance]{prop}{cssphamming}
\label{prop:cssp_hamming}
    There exists a quantum algorithm for solving the CSSP under the Hamming distance, with runtime $\mathcal{\tilde{O}}\Bigl(
    \sigma^{d(\log(d)+2)/2} $ $ \cdot (kn)^{\log(d)/2} $ $ \cdot k^2 n^3
    \Bigr)$.
\end{restatable}

The best known non-trivial classical algorithms for the CSP run in time $\bigo{\sigma \cdot 2^k \cdot n^{2k}}$ \cite{nicolas_hardness_2005}, $\bigo{kn + k \cdot d^d}$~\cite{gramm_fixed-parameter_2003}, or $\bigo{kn + kd\cdot M^d}$ for $M = \bigo{\sigma}$~\cite{ma_more_2010,wang_efficient_2009, chen_fast_2011}, with each suited to different parameter regimes.
\Cref{prop:csp_hamming} therefore represents a quadratic improvement in the dominant scaling factor for the CSP under Hamming distance when the alphabet size is unbounded.
\Cref{prop:csp_levenshtein} is a non-trivial speedup, but the speedup is not polynomial in the size of the input.
Finally, \Cref{prop:csp_trivial} is a speedup over trivial classical searches, which are the best methods when $k$ and $d$ are both very large.

We also show that for the special case of binary alphabets,~\Cref{prop:csp_hamming} is tight under the quantum analogue of the Strong Exponential Time Hypothesis, QSETH, formulated by Buhrman et al.~\cite{buhrman_framework_2021}.

The best classical algorithm for the CSSP runs in time $\bigo{\sigma^{d(\log(d) + 2)} (kn)^{\log(d)}
k^3 n^4}$~\cite{marx_closest_2005}.
Similarly to~\Cref{prop:csp_hamming}, \Cref{prop:cssp_hamming} presents a quadratic improvement in the dominant scaling factors.

\section{Outlook}
\label{sec:outlook}
This work presents a suite of parameterized quantum algorithms for the Closest String problem, each achieving quadratic speedups over classical counterparts.
By incorporating known subroutines -- Grover search~\cite{grover_fast_1996}, quantum random walks~\cite{magniez_search_2011}, and quantum dynamic programming~\cite{ambainis_quantum_2018} -- into our algorithmic design, we demonstrate that quantum computing can offer meaningful advantages for hard combinatorial problems.

The CSP is particularly relevant in bioinformatics, where it serves as a computationally tractable proxy for full multiple sequence alignment. 
Improvements in solving this problem could lead to more efficient preprocessing, filtering, or heuristic evaluation stages in large-scale comparative genomics and phylogenetic analysis pipelines. 
As quantum hardware matures, our algorithms may form the basis for practical tools in genomics and related fields, especially in scenarios where fast, exact solutions are needed for moderate-size instances.

Several promising directions remain open for future work. 
There are many other classical algorithms for the CSP, targeting other combinations of parameters~\cite{amir_configurations_2012, ma_more_2010,wang_efficient_2009, chen_fast_2011}.
Adapting these algorithms to the quantum regime may yield improved speedups or tighter parameter dependencies. 
Related optimisation problems, such as the Distinguishing String Selection problem~\cite{kevin_lanctot_distinguishing_2003}, share structural features with the CSP and may benefit from similar quantum treatment.
More broadly, exploring the interplay between parameterized complexity theory and quantum algorithm design continues to be a rich area for foundational and applied research.

\section{Notation and Preliminaries}
\label{sec:notation}
\begin{table}[H]
\begin{tabularx}{\textwidth}{@{}p{0.15\textwidth}X@{}}
\toprule
\multicolumn{2}{@{}l}{\textbf{General Mathematical Notation}}\\
$\mathbb{Z}_+$  & The set of non-negative integers.\\
$\mathbb{R}_+$ & The set of non-negative real numbers.\\
$[M]$ & The set $\{0, 1, \ldots, M-1\}$.\\
\textbf{Symbols}\\
$\Sigma$  & An alphabet, usually $\{$A, C, G, T$\}$. \\
$\Sigma'$ & An extended alphabet with a ``gap'' symbol, $\Sigma' \coloneqq \Sigma \cup \{-\}$.\\
$\Sigma^*$  & The set of strings of any length with symbols from $\Sigma$, $\Sigma^* \coloneqq \bigcup_{l \in \mathbb{Z}_+} \Sigma^l$.\\
$\Sigma^{\leq n}$  & The set of strings of length at most $n$ with symbols from $\Sigma$, $\Sigma^{\leq n} \coloneqq \bigcup_{l =0}^{n} \Sigma^l$.\\
\textbf{Strings}\\
$s$ & A string $s \in \Sigma^*$.\\
$s[i]$  & The $i$th character of $s$.\\
$s[i, j]$   & The substring of $s$ running from $i$ to $j$ inclusive.\\
$\abs{s}$   & The length of $s$.\\
\textbf{Distances}\\
$d(\cdot,\,\cdot)$  & A generic distance metric $d: \Sigma^* \times \Sigma^* \rightarrow \mathbb{R}_+$.\\
$d_H(\cdot,\,\cdot)$  & The Hamming distance $d_H: \bigcup_{n \in \mathbb{Z}_+} \left(\Sigma^n \times \Sigma^n\right) \rightarrow \mathbb{Z}_+$.\\
$d_L(\cdot,\,\cdot)$  & The Levenshtein distance $d_L: \Sigma^* \times \Sigma^* \rightarrow \mathbb{Z}_+$.\\
$d_\delta(\cdot,\cdot)$ & A general weighted edit distance $d_\delta(\cdot,\cdot): \Sigma^* \times \Sigma^* \rightarrow \mathbb{R}_+$.\\
\multicolumn{2}{@{}l}{\textbf{Parameters}}\\
$k$ & Number of strings in problem input. \\
$d$ & Threshold distance in problem output. \\
\multicolumn{2}{@{}l}{\textbf{Markov Chains}}\\
$P$ & A Markov chain with transition matrix $P$.\\
$\Omega$ & The set of states of a chain $P$.\\
$\pi$ & The stationary distribution of $P$.\\
$\delta(P)$ & The eigenvalue gap of $P$. For lazy, reversible $P$, $\delta(P) = 1 - \lambda_2$.\\
\bottomrule
\end{tabularx}
\caption{A summary of notation used in this work.}
\end{table}

\subsection{Strings and distances}
A string $s$ is a finite sequence of letters of some finite alphabet $\Sigma$. For bioinformatics applications, this will often be the size-4 alphabet $\Sigma = \{A,\,C,\,G,\,T\}$. Let $\sigma = \abs{\Sigma}$.
When insertions or deletions are allowed in string transformations, we extend the alphabet to include a ``gap symbol'', defining $\Sigma' = \Sigma \cup \{-\}$.
The length of $s$ is $\abs{s}$, and for $0 \leq i < \abs{s}$, its $i$th character is $s[i]$. For $0 \leq i < j < \abs{s}$, $s[i]s[i+1]\ldots s[j]$ is a \emph{substring} of $s$ and is denoted by $s[i,j]$.

The Hamming distance $d_H(\cdot,\cdot)$ between two strings $s$ and $s'$ of equal length is the number of positions $i$ for which $s[i] \neq s'[i]$.
The Levenshtein distance~\cite{levenshtein_binary_1965} $d_L(\cdot,\cdot)$ between two strings $s$ and $s'$ is the minimum number of characters that must be inserted, deleted, or substituted to transform $s$ into $s'$. For a formal definition based on the notion of alignments, see, e.g., Chapter 6 of~\cite{pevzner_computational_2000}. 
One can also consider a more general scenario, in which the cost of inserting, deleting and substituting can vary depending on which letter(s) of the alphabet are involved. 
Let $\delta: \Sigma' \times \Sigma' \rightarrow \mathbb{R}_+$ be a mapping that encodes these costs.
If $\delta$ satisfies the conditions of a \emph{metric}~\cite{nicolas_hardness_2005}, then the corresponding function encoding the cost of transforming $s$ to $s'$ is a true distance.
In this case, we will use the notation  $d_\delta(\cdot,\cdot)$ and refer to it as a weighted edit distance.

\subsection{Parameterized Complexity} Parameterized Complexity is a branch of complexity theory that classifies problems based on several aspects of their input or output, rather than solely on the size of the input.
It is most commonly applied in the context of $\NP$-hard problems.
Whereas classical complexity considers only input size $n$, we now consider some additional parameter $k$ that represents additional structure in the problem, such as the number of input strings in a string selection problem, and study the behaviour of the problem with both $n$ and $k$.

A key goal is to determine whether a problem can be solved in time polynomial in $n$ for any fixed value of the parameter $k$. 
In this case, the problem is tractable for small values of $k$, even if it remains $\NP$-hard in general.
If an algorithm exists with runtime $\bigo{n^{c} \cdot f(k)}$ for a constant $c$ and some, generally exponential, function $f$, we say the problem is \emph{Fixed Parameter Tractable}. 
The complexity class of such problems is $\FPT$; $\FPT$ contains all of $\P$ as well as any $\NP$ optimisation problem that admits an efficient poly-time approximation scheme.
If, instead, an algorithm exists with runtime $n^{g(k)} \cdot f(k)$ for some computable $g(k)$, the problem is said to lie in $\XP$. 
Clearly $\FPT \subseteq \XP$, and in fact the inclusion is strict by diagonalization.
These definitions naturally extend to algorithms with several parameters.

\subsection{String Selection Problems} 
String selection problems~\cite{kevin_lanctot_distinguishing_2003} are a collection of interrelated problems that reduce to finding a pattern that occurs, possibly with some error, in a set of strings.
The size of the error must be calculated using some distance measure.
The most commonly studied in the literature is the less computationally intensive Hamming distance, but some works have considered the Levenshtein distance or general weighted edit distances.
Different distances can change the complexity class of the problem, so we will be precise about which one is under consideration throughout.

The most straightforward is the CSP: given $k$ strings, find a string within distance $d$ of all of them.
We state this formally as follows:
\begin{prob}[Closest String Problem]
    Given $k$ strings $s_1,\,\ldots,\,s_k \in \Sigma^{\leq n}$ and a threshold distance $d$, find a string $s \in \Sigma^*$ such that $\max_i\{d(s,\,s_i)\} \leq d$.
\end{prob}
The CSP is also known as the \emph{Center String Problem}.

When $d$ is the Hamming distance, we require all strings to be of equal length for the problem to be well-defined.
The CSP is known to be $\NP$-complete, even for binary alphabets under the Hamming distance~\cite{frances_covering_1997}.
There are several known parameterized algorithms for the CSP under various distance metrics. A particularly simple one is due to Nicolas and Rivals~\cite{nicolas_hardness_2005}, and it functions for any ``natural'' distance metric where deletion of a letter is always strictly cheaper than substituting to a different letter and then deleting the new one.
They proceed by Dynamic Programming, and achieve a runtime of $\bigo{\sigma\cdot 2^k \cdot n^{2k}}$.
This places the CSP in $\XP$ with respect to $k$. Note that $d$ is allowed to be unbounded in this case.

Another notable algorithm is due to Gramm et al.~\cite{gramm_fixed-parameter_2003}, solving the CSP under the Hamming distance. For fixed $d$, the algorithm sets up a small space of feasible solutions and searches over them in time $\bigo{kn + k \cdot d^{d+1}}$. This places the CSP under the Hamming distance in $\FPT$ with respect to the distance threshold $d$.
Other algorithms~\cite{ma_more_2010,wang_efficient_2009,chen_fast_2011} achieve $\bigo{kn + kd\cdot M^d}$ for $M = \bigo{\sigma}$. For fixed alphabet size, these asymptotically outperform the algorithm of Gramm et al. 

A close relative of the CSP is the CSSP:
\begin{prob}[Closest Substring Problem]
    Given $k$ strings $s_1,\,\ldots,\,s_k \in \Sigma^{\leq n}$ and a threshold distance $d$, find a length-$L$ string $s \in \Sigma^L$ such that for all $i = 1,\ldots,k$, there exists a substring $w_i$ of $s_i$ with $d(s, w_i) \leq d$.
\end{prob}
The CSP is a special case of the CSSP with $L = n$.
When $d$ is the Hamming distance, we require that the substrings $w_i$ have length $L$.
Marx~\cite{marx_closest_2005} presents an algorithm for the CSSP under Hamming distance with runtime $\bigo{\sigma^{d(\log(d) + 2)} (kn)^{\log(d)}
k^3 n^4}$, placing CSSP in $\XP$ for this distance.
The algorithm proceeds by arguing about so-called \emph{generators} of length-$L$ substrings.

\subsection{Quantum Algorithms} 
For some function $f: [M] \rightarrow [N]$, $M,N \in \mathbb{N}$, we refer to a unitary gate $O_f$ that implements $\ket{x}\ket{b} \mapsto \ket{x}\ket{b + f(x) \mod N}$ for $x \in [M], b \in [N]$ as an oracle for $f$.

We assume oracle access model to a string $s \in \Sigma^n$, wherein we may query a quantum unitary gate $O_s$ implementing 
\begin{equation}
    O_s : \ket{i}\ket{b} \mapsto \ket{i}\ket{b + s[i] \mod \sigma}
\end{equation} 
for any $i \in [n]$ and $b \in [\sigma]$. 
This quantum query model~\cite{ambainis_quantum_2004, buhrman_complexity_2002} is standard in the literature of quantum algorithms.

A key quantum subroutine is the Grover search algorithm, which finds a marked element of a set of $n$ elements in time $\bigotilde{\sqrt{n}}$:
\begin{theorem}[Grover search~\cite{grover_fast_1996}]
\label{thm:grover}
    There is a quantum algorithm that, given a function $f: [n] \rightarrow \{0,\,1\}$, finds an index $i$ with $f(i) = 1$ or reports that no such index exists.
    The algorithm uses $\bigotilde{\sqrt{n}}$ queries to an oracle for $f$.
\end{theorem}
A use of this algorithm to find the minimum (equivalently, maximum) of a set of $n$ values in $\bigotilde{\sqrt{n}}$ was proposed by Durr and Hoyer~\cite{durr_quantum_1999}:
\begin{theorem}[Quantum minimum/maximum finding~\cite{durr_quantum_1999}]
    There is a quantum algorithm that, given a function $f: [n] \rightarrow \mathbb{R}$, finds the index $i$ such that $f(i) \leq f(j) \ \forall j$ (equivalently, $f(i) \geq f(j) \ \forall j$).
    The algorithm uses $\bigotilde{\sqrt{n}}$ queries to an oracle for $f$.
\end{theorem}

A second quantum search routine is due to Magniez et al.~\cite{magniez_search_2011} and proceeds via a random walk.
It is a stronger version of the algorithms provided by Ambainis~\cite{ambainis_quantum_2014} and Szegedy~\cite{szegedy_quantum_2004}.
A full introduction to this result, including the requisite background on Markov chains, is given in~\Cref{app:quantum_search}, but we briefly state the main results here.

Suppose we have a set of states $\Omega$ and a subset of $M \subset \Omega$ of \emph{marked} states, and we wish to find some $x \in M$.
Classically, we might consider a search algorithm as described by the pseudocode in~\Cref{code:c_search}.
Assuming that the algorithm maintains a data structure $D$ that associates data $D(x)$ to each state $x \in \Omega$, then the complexity of the algorithm rests on 3 quantities:
\begin{itemize}
    \item \textbf{Setup} $S$: The cost of sampling $x \in \Omega$ according to the distribution $s$ and constructing $D(x)$ from $x$.
    \item \textbf{Update} $U$: The cost of simulating a transition $x \rightarrow y$ according to $P$ and of updating $D(x)$ to $D(y)$.
    \item \textbf{Checking} $C$: The cost of checking whether $x \in M$ using $D(x)$.
\end{itemize}
Then there is a classical search algorithm as described by~\Cref{thm:c_search}:
\begin{theorem}[Classical Search via Random Walk~\cite{magniez_search_2011}]
    \label{thm:c_search}
    Let $\delta > 0$ be the eigenvalue gap of an ergodic, symmetric Markov chain. Let ${\abs{M}}/{\abs{\Omega}} \geq \epsilon > 0$ whenever $M$ is non-empty. 
    For the uniform initial distribution $s$, the above search algorithm succeeds w.h.p. if $t_1 = \mathcal{O}\left(\frac{1}{\delta}\right)$ and $t_2 = \mathcal{O}\left(\frac{1}{\epsilon}\right)$. 
    The incurred cost is of order:
    \begin{equation}
            S + \frac{1}{\epsilon}\left( \frac{1}{\delta} U + C \right).
    \end{equation}
\end{theorem}
\begin{algorithm}
\caption{Classical Search via Random Walk~\cite{magniez_search_2011}}\label{code:c_search}
\begin{algorithmic}[1]
\Require Integers $t_1, t_2$, probability distribution $s$ over $\Omega$, Markov chain $P$ on $\Omega$.
\State $x \gets$ sample from $s$
\For{$i = 1, \ldots, t_2$} 
    \If{$x \in M$}
        \State \Return $x$
    \Else
        \State $x \gets t_1$th step of a simulation of $P$ starting from $x$
    \EndIf
\EndFor
\State \Return "No marked state exists"
\end{algorithmic}
\end{algorithm}
For the quantum case, the data structure is captured by states belonging to an expanded Hilbert space $\mathbb{C}^{\Omega_D}$ where $\Omega_D = \{ (x, D(x)) : x \in \Omega \}$.
For conciseness, we write $\ket{x}_D \equiv \ket{x, D(x)}$.
For the search, we actually consider walks on the edges of the Markov chain, so the walk lives in the Hilbert space $\mathcal{H} = \mathbb{C}^{\Omega_D \times \Omega_D}$.

Suppose we now have a Markov chain $P$ with stationary distribution $\pi$ and transition probabilities $\{p_{xy} : x,y \in \Omega\}$.
Let $P^*$ be the time-reversed Markov chain with transition probabilities $p_{xy}^*$.
The costs of the quantum search are given by:
\begin{itemize}
    \item \textbf{(Quantum) Setup} $S$: The cost of constructing the state $\sum_x \sqrt{\pi_x} \ket{x}_D \ket{0}_D$.
    \item \textbf{(Quantum) Update} $U$: The cost of realising any of the unitary transformations
    \begin{align*}
        \ket{x}_D \ket{0}_D \mapsto \ket{x}_D \sum_y \sqrt{p_{xy}} \ket{y}_D, \\
        \ket{0}_D \ket{y}_D \mapsto \sum_x \sqrt{p^*_{yx}}\ket{x}_D \ket{y}_D,
    \end{align*}
    or their inverses.
    \item \textbf{(Quantum) Checking} $C$: The cost of realising the conditional phase flip
    $$
    \ket{x}_D \ket{y}_D \mapsto \begin{cases}
        -\ket{x}_D \ket{y}_D & \text{if } x \in M,\\
        \phantom{-}\ket{x}_D \ket{y}_D & \text{otherwise.}
    \end{cases}
    $$
\end{itemize}
Then there exists a quantum search algorithm as described by~\Cref{thm:q_search}.
\begin{theorem}[Quantum Search via Random Walk~\cite{magniez_search_2011}]
\label{thm:q_search}
Let $S$, $U$ and $C$ respectively be the costs of setting up the required data structure, updating the data structure at each step of the walk, and applying a phase flip to elements $x \in M$ that are marked for the search.

Let $\delta > 0$ be the eigenvalue gap of a reversible, ergodic Markov chain $P$, and let $\epsilon > 0$ be a lower bound on the probability that an element chosen from the stationary distribution of $P$ is marked whenever $M$ is non-empty. Then, there is a quantum algorithm that with high probability, determines if $M$ is empty or finds an element of $M$, with cost of order 
\begin{equation}
    S + \frac{1}{\sqrt{\epsilon}} \left( \frac{1}{\sqrt{\delta}} U + C\right).
\end{equation}
\end{theorem}

A further quantum routine we make use of is Quantum Dynamic Programming. 
First introduced by Ambainis et al.~\cite{ambainis_quantum_2018} on the Boolean hypercube, it was extended by Glos et al.~\cite{glos_quantum_2021} to the $k$-dimensional lattice graph with vertices in $\{1,2,\ldots,n\}^k$.
As with classical Dynamic Programming, the algorithm is applicable to problems that have the \emph{optimal substructure} property: optimal solutions can be constructed from optimal solutions of sub-problems.
The algorithm is presented for a ``path in the hyperlattice'' problem, but it applies equally to any classical dynamic programming, replacing the Grover search with quantum minimum/maximum finding as necessary.
The quantum algorithm is recursive and uses nested quantum searches at each layer of recursion.
While the classical query complexity is $\tilde{\Theta}\left(n^k\right)$, the quantum algorithm has complexity $\bigotilde{T(n)^k}$ for some function $T(n)$ satisfying $n/e \leq T(n) < n$ for every $n$.
Although this is not a polynomial speedup in the size of the problem, it still represents a non-trivial improvement in the complexity of the algorithm.
We state this result in the following Theorem.
\begin{theorem}[Quantum Dynamic Programming~\cite{glos_quantum_2021}]
\label{thm:qdp}
    There is a quantum algorithm that solves a problem with optimal substructure on a $k$-dimensional lattice with vertices in $\{1,2,\ldots,n\}^k$ in time $\bigotilde{T(n)^k}$ for some $T(n) < n + 1$.
\end{theorem}

Finally, we will often rely on performing classical computing as a subroutine.
In particular, we will want to implement oracles $O_f$ for non-trivial classical functions $f$.
It was shown by Bennett~\cite{bennett_logical_1973} that deterministic classical, irreversible algorithms may be efficiently converted to reversible ones.
This was extended by Zuliani~\cite{zuliani_logical_2001} to include nondeterministic and probabilistic computations.
This general-purpose result allows for the efficient simulation of classical computation on a quantum computer.

\section{Quantum Algorithms for the Closest String Problem}
\label{sec:q_csp}
We present 3 algorithms for the CSP.
The first is a straightforward nested application of quantum minimum/maximum finding, and is applicable for any desired distance metric.
It is suited to the case where there are very many input strings, each of which is not too long. 
This may be the case, for example, when seeking a conserved fragment of a gene among all the different strains of pathogens such as SARS-CoV-2. 

The second is a quantisation of a slightly modified version of the algorithm due to Gramm et al.~\cite{gramm_fixed-parameter_2003} for CSP under the Hamming distance.
It is best suited to the setting where all the input strings are closely related, such as representing the DNA data of individuals from the same species.
Interestingly, for a fixed distance threshold, the runtime scales only linearly with the length of input strings and number of input strings.

Finally, the third algorithm is a quantisation of the dynamic programming algorithm due to Nicolas and Rivals~\cite{nicolas_hardness_2005} using the results of Glos et al.~\cite{glos_quantum_2021}.
It is suited to cases where there are few input strings and no bound on the threshold distance is known -- at present, this is most likely of theoretical interest only.

\subsection{A Trivial Quantum Search Algorithm}
Perhaps the most straightforward method to solve the CSP is a nested search algorithm.
The outer loop is a quantum minimum finding that runs over all length-$n$ strings $s \in \Sigma^n$. The function to be minimised is the maximum distance between $s$ and any of the input strings $s_1,\ldots,s_k$.
Instead of checking each of these distances sequentially, we use a quantum maximum finding.
For any $s$, the search runs over indices $i = 1,\ldots, k$ with search function the distance between $s$ and $s_i$.

Formally, we may write the search functions:
\begin{alignat}{7}
    f: && \ \Sigma^{n} &\rightarrow \mathbb{Z}_+ && && \qquad \qquad g_s:& \ [k] &\rightarrow \mathbb{Z}_+, \\
    && s &\xrightarrow{} \max_{i \in [k]} d(s,\,s_i) &&&&
    &i &\xrightarrow{} d(s_,\,s_i) . \nonumber
\end{alignat}

The inner search of $g_s$ runs in time $T_I = \bigotilde{\sqrt{k} \cdot T_d}$, where $T_d$ is the time required to compute $d(s,s')$ for any length-$n$ strings $s,\,s'$.
In particular, $T_d = \bigo{n}$ for the Hamming distance and $T_d = \bigo{n^2}$ for the Levenshtein distance.
The outer search of $f$ runs in time $\bigotilde{\sqrt{\sigma^n} \cdot T_I} = \bigotilde{\sigma^{n/2} \cdot \sqrt{k} \cdot T_d}$. 
We therefore obtain our first result,~\Cref{prop:csp_trivial}:
\csptrivial*

The runtime will usually be dominated by the search over all length-$n$ strings.
However, it has very favourable scaling with respect to $k$, and could therefore be feasible in very high data volume scenarios.

In the classical setting, it can be shown that exhaustive search is essentially optimal for the CSP under Hamming distance in the binary case~\cite{abboud_can_2023}, assuming SETH holds.
The proof follows via a reduction from $(q,\,K)$-CNF  SAT, i.e. the problem of finding satisfying assignments to $q$-ary CNF formulae, to the \emph{Remotest String Problem} (RSP), a natural analogue of the CSP where the goal is to find a string far from all inputs.
For binary alphabets, the CSP and the RSP are equivalent by flipping all bit strings in the inputs.
We restate the result in~\Cref{thm:reduction}:

\begin{theorem}[Based on Theorems 1.1 \& 1.2 from~\cite{abboud_can_2023}]
    \label{thm:reduction}
    An algorithm that solves the CSP for binary alphabets under Hamming distance in time $\bigo{2^{(1-\epsilon)n}\poly(k)}$ implies the existence of an algorithm for $(q,\,K)$-SAT on $N$ variables with $M$ clauses which runs in time $\bigo{2^{(1-\epsilon/2)N} \poly(M)}$.
\end{theorem}

The framework of Buhrman et al.~\cite{buhrman_framework_2021} allows us to make the \emph{Basic $(q,K)$-CNF SAT QSETH} conjecture:
\begin{conjecture}[Basic $(q,K)$-CNF SAT QSETH]
    A quantum algorithm cannot, given an input $C$ from the set of $(q,K)$-CNF formulae on $N$ variables and $M$ clauses, decide whether there exists $x \in [q]^{N}$ such that $C(x) = 1$ in time $\bigo{q^{\frac{N}{2}(1-\delta)} \poly(M)}$, for any $\delta > 0$. 
\end{conjecture}

Using the same reduction that leads to~\Cref{thm:reduction}, we see that~\Cref{prop:csp_trivial} is tight for binary alphabets unless Basic $(q,K)$-CNF SAT QSETH is false:

\begin{theorem}
    No quantum algorithm solves the binary alphabet CSP in time $\bigo{2^{(1-\epsilon)\frac{n}{2}} \poly(k)}$ unless Basic $(q,K)$-CNF SAT QSETH is false.
\end{theorem}

We must therefore look to parameterized algorithms to achieve any meaningful speedups.

\subsection{A Quantum Random Walk Algorithm for the CSP under Hamming Distance}
We present a quantisation of the algorithm due to Gramm et al.~\cite{gramm_fixed-parameter_2003}, which we briefly describe here.
The algorithm is specifically designed for the setting where the threshold distance $d$ is small.

As a pre-processing step, note that any position where all $k$ strings agree can be removed.
If more than $kd$ positions are remaining, then no solution can exist, and the algorithm terminates.
This pre-processing takes time $\bigo{kn}$, which is usually negligible.

The core algorithm is based on the simple observation that if a solution $\hat{s}$ exists with a threshold distance $d$, then certainly $d_H(\hat{s}, s_1) \leq d$.
Immediately, the search space is reduced to a ball of radius $d$ around $s_1$.
Rather than explore the whole space exhaustively, we consider iteratively modifying a candidate string $s$, which is initialised as $s = s_1$.
At each step, we find some minimal index $i^*$ for which $d_H(s, s_{i^*}) > d$.
We then consider the first $d+1$ positions where $s\neq s_{i^*}$, and for each such position, we form a new candidate string by replacing the character in $s$ with the character in $s_{i^*}$.
Since, by hypothesis, we have $d_H(\hat{s}, s_{i^*}) \leq d$, there are at most $d$ positions where both $s \neq s_{i^*}$ and $\hat{s} \neq s_{i^*}$.
Therefore, at least one of the modifications must decrease $d_H(s, \hat{s})$, therefore moving the candidate solution closer to the true solution.
After $d$ such rounds, a solution must be found if one exists.

The algorithm can be formulated recursively, with at most $d$ layers of recursion.
As a further optimization, we note that if, at any layer, the candidate solution's distance is strictly greater than $d$ plus the number of rounds remaining, then no solution is possible.
The algorithm is described by the pseudocode in~\Cref{code:treesearch}.
The inputs at the first level of recursion are $s = s_1$ and $\Delta d = d$.
\algrenewcommand\algorithmicensure{\textbf{Output:}}
\begin{algorithm}
\caption{\texttt{TreeSearch}$(s,\,\Delta d)$~\cite{gramm_fixed-parameter_2003}}\label{code:treesearch}
\begin{algorithmic}[1]
\Require \textbf{(Global Variables)} $s_1,\,\ldots,\,s_k \in \Sigma^{kd}$, threshold distance $d$.
\Require Candidate string $s$ and integer $\Delta d$.
\Ensure String $\hat{s}$ with $\max_{i}d_H(\hat{s},\,s_i) \leq d$ and $d_H(\hat{s},\,s) \leq \Delta d$, if one exists.
\If{$\Delta d < 0$}
    \State \Return ``Not found''
\EndIf
\If{$d_H(s,\,s_i) > d + \Delta d$ for some $i \in [1\,..\,k]$}
    \State \Return ``Not found''
\EndIf
\If{$d_H(s,\,s_i) \leq d$ for all $i \in [1\,..\,k]$}
    \State \Return $s$
\EndIf
\State $i^* \gets \min \{ i : d_H(s,\,s_i) > d \}$
\State $P \gets \{ p : s[p] \neq s_{i^*}[p] \}$
\State $P' \gets P[1\,..\,d+1]$
\ForAll{$p \in P'$}
    \State $s' \gets s$
    \State $s'[p] \gets s_{i^*}[p]$
    \State $s_{\text{new}} \gets \texttt{TreeSearch}(s',\,\Delta d - 1)$
    \If{$s_{\text{new}} \neq$ ``Not found''}
        \State \Return $s_{\text{new}}$
    \EndIf
\EndFor
\State \Return ``Not found''
\end{algorithmic}
\end{algorithm}

The algorithm, in the worst case, searches over a \emph{perfect $(d+1)$-ary tree} $\mathcal{T}$: a tree where all interior nodes have $d+1$ children and all leaves are at the same depth. 
The depth of the tree is $d+1$.
We therefore call the algorithm \texttt{TreeSearch}.
Gramm et al. compute the runtime of the algorithm, which we restate in~\Cref{thm:treesearch}.
\begin{theorem}[\cite{gramm_fixed-parameter_2003}]
\label{thm:treesearch}
   \texttt{TreeSearch}, described in~\Cref{code:treesearch}, with (pre-processed) global variables $s_1,\ldots, s_k$, threshold distance $d$ and inputs $s = s_1$ and $\Delta d = d$ solves the CSP under the Hamming distance in time $\bigo{k\cdot d^{d+1}}$.
\end{theorem}
We now describe a quantum algorithm based on \texttt{TreeSearch} and the quantum search via random walk~\cite{magniez_search_2011}.
To use the quantum search, we first need to define a Markov chain $P$.
Let $\Omega$ be the set of nodes of $\mathcal{T}$.
We may denote some node of $\mathcal{T}$ by a $d$-dimensional, integer-valued vector $\mathbf{j} \in \{0,\ldots,d+1\}^{d}$.
The root is $(0,\ldots,0)$; for $\mathbf{j}$ in the $l^{\text{th}}$ layer, $1 \leq l \leq d$, the first $l$ coordinates are non-zero, corresponding to which branch was taken at layer $l$.
Then we can write:
\begin{equation}
    \Omega = \{ \mathbf{j} = (j_1,\ldots,j_l, 0,\ldots, 0) \in \{0,\ldots,d+1\}^{d} : 0 \leq l \leq d; 1 \leq j_i \leq d \}.
\end{equation}

We define a very symmetric chain $P$ on $\Omega$ as follows:
\begin{align}
\text{If $\mathbf{j}$ is neither root nor leaf,} \quad
    p_{\mathbf{jj'}} &= \begin{cases}
    \frac{1}{2} & \text{if } \mathbf{j'} = \mathbf{j}, \\
    \frac{1}{4(d+1)} & \text{if $\mathbf{j'}$ is a child of $\mathbf{j}$}, \\  
    \frac{1}{4} & \text{if $\mathbf{j'}$ is the parent of $\mathbf{j}$}, \\
    0 & \text{otherwise.}
\end{cases} \\
\text{If $\mathbf{j}$ is a leaf,} \quad
p_{\mathbf{jj'}} &= \begin{cases}
    \frac{1}{2} & \text{if } \mathbf{j'} = \mathbf{j}, \\
    \frac{1}{2} & \text{if $\mathbf{j'}$ is the parent of $\mathbf{j}$}, \\
    0 & \text{otherwise.}
\end{cases} \\
\text{If $\mathbf{j}$ is the root,} \quad
p_{\mathbf{jj'}} &= \begin{cases}
    \frac{1}{2} & \text{if } \mathbf{j'} = \mathbf{j}, \\
    \frac{1}{2(d+1)} & \text{if $\mathbf{j'}$ is a child of $\mathbf{j}$}, \\
    0 & \text{otherwise.}
\end{cases}
\end{align}
The set of edges of $P$ is
$E = \{(\mathbf{j}, \mathbf{j'}) \in \Omega \times \Omega: \mathbf{j} \text{ in layer }l,\ \mathbf{j'} \text{ in layer }l+1,\ j_i = j'_i, i = 1,\ldots, l\}$.

We show in~\Cref{app:markov} that the eigenvalue gap of $P$ is at least inverse-quadratic in $d$:
\begin{lemma}
    The eigenvalue gap $\delta(P)$ of $P$ is lower-bounded by $\delta(P) > {2^{-9} d^{-2}}$.
\end{lemma}
\begin{proof}
    See~\Cref{app:markov}.
\end{proof}

We now define the data structure and compute the costs of setup, update and checking required by the quantum search routine.

\textbf{Data Structure.}
By analysing~\Cref{code:treesearch}, we see that the data we need to track is the history of modifications made to the candidate string.
In particular, given a list of indices $i_1,\ldots,i_l \in \{1,\ldots,k\}$ representing the donor strings, and positions $k_1,\ldots,k_l \in \{1,\ldots,kd\}$ representing the character position that was substituted, we can build the candidate string $s$ from $s_1$ in time $\bigo{l} = \bigo{d}$.
This suggests a simple data structure for our algorithm.
For $\mathbf{j} = (j_1,\,\ldots,\,j_l,\,0,\,\ldots,\,0)$, we take $D(\mathbf{j}) = \{ (i_1,\,\ldots,\,i_l,\,0,\,\ldots,\,0),\ (k_1,\,\ldots,\,k_l,\,0,\,\ldots,\,0) \}$.
This structure requires $\bigo{d\log(kd)}$ qubits.

\textbf{Setup.}
First, we prepare the state 
$\ket{\Phi} = \sum_{\mathbf{j} \in \Omega} \sqrt{\pi_{\mathbf{j}}} \ket{\mathbf{j}}$,
which is a state-preparation operation on $\bigo{d\log(d)}$ qubits. By, e.g., the results of Zheng et al.~\cite{zhang_quantum_2022}, this can be done in time $\bigo{d\log(d)}$.
We then need to implement the transformation:
\begin{equation}
    U_D : \ket{\mathbf{j}} \ket{\mathbf{0}} \mapsto \ket{\mathbf{j}}\ket{D(\mathbf{j})}.
\end{equation}
In the worst case, this requires computing $d$ modifications to the candidate string $s$.
For each modification, we need to compare $s$ to the $k$ input strings, which takes time $\bigo{k\cdot kd}$.
So the worst case complexity of $U_D$ is $\bigotilde{k^2d^2}$.
Finally, we append the state $\ket{0}_D \equiv \ket{0}\ket{0}$ at no cost.
So $S = \bigotilde{k^2d^2}$.

\textbf{Update.}
We wish to implement the transformation:
\begin{equation}
    \ket{\mathbf{j}}_D \ket{\mathbf{0}}_D \mapsto \ket{\mathbf{j}}_D \sum_\mathbf{j'} \sqrt{p_{\mathbf{j}\mathbf{j'}}} \ket{\mathbf{j'}}_D.
\end{equation}
We first implement a diffusion-like operator that realises
\begin{equation}
\ket{\mathbf{j}}_D\ket{\mathbf{0}}_D \mapsto \ket{\mathbf{j}}_D \sum_{\mathbf{j'}}\sqrt{p_{\mathbf{j}\mathbf{j'}}} \ket{\mathbf{j'}}\ket{\mathbf{0}}.
\end{equation}
This controlled-state-preparation on $\bigo{d\log(d)}$ qubits can again be performed in time $\bigotilde{d\log(d)}$.
To update the data structure, we must compute $\ket{\mathbf{j'}}_D$ from $\ket{\mathbf{j}}_D$, where $\mathbf{j'}$ is either the parent or child of $\mathbf{j}$.
To go to the parent, we simply set the last entries $i_l, k_l$ to zero in $\bigo{1}$ depth.
To go to a child, we must compute one further modification to the candidate string, in time $\bigotilde{k^2d}$. So $U = \bigotilde{k^2d}$.

\textbf{Checking.}
It takes $\bigo{d}$ time to compute the candidate string and $\bigo{k^2d}$ time to check it against all the input strings. So $C = \bigotilde{k^2d}$.

\textbf{Overall time complexity.}
Finally, we need to bound the probability of sampling, under the stationary distribution, a marked element.
The worst case is when a single leaf is marked.
From the expression for $\tilde{\pi}$ derived in~\Cref{eq:stationary_distribution}, we may bound:
\begin{equation}
    \epsilon \coloneqq \mathbb{P}_{\mathbf{J}\sim \pi}(\mathbf{J} \in M) \geq \frac{\tilde{\pi}_d}{(d+1)^d} = \frac{1}{2d(d+1)^d}.
\end{equation}
Substituting the quantum costs as well as the computed values of $\delta$ and $\epsilon$ into the equation provided by~\Cref{thm:q_search} gives a runtime of order:
\begin{equation}
    \bigotilde{k^2d^2 + \sqrt{2d (d+1)^d} \left(\sqrt{2^9 d^2} k^2d + k^2d \right)}
    = \bigotilde{k^2 \cdot d^{(d+5)/2}}.
\end{equation}

We therefore obtain our second result,~\Cref{prop:csp_hamming}.
\csphamming*

Recall that the classical algorithm has runtime $\bigo{kn + k \cdot d^{d+1}}$.
We thus obtain a quadratic improvement in the dominant factor $d^d \rightarrow d^{d/2}$.

\subsection{A Quantum Dynamic Programming Algorithm for the CSP under Levenshtein Distance}
Nicolas and Rivals~\cite{nicolas_hardness_2005} present a Dynamic Programming algorithm for the CSP under any natural weighted edit distance.
For clarity, we restrict here to the Levenshtein distance.
It runs on a $2k$-dimensional Boolean array, where each dimension has length $n+1$.
This is exactly the setting considered by Glos et al.~\cite{glos_quantum_2021}.
The only change required is to replace the $\bigo{1}$-cost membership queries of Glos et al. with $\bigo{\sigma \cdot 2^k}$-cost logical ORs required by Nicolas and Rivals.
By applying~\Cref{thm:qdp}, we obtain a quantum algorithm with runtime $\bigotilde{\sigma \cdot 2^k \cdot T(n)^{2k}}$, and thus our third result,~\Cref{prop:csp_levenshtein}.
\csplevenshtein*

\section{Quantum Algorithms for the Closest Substring Problem}
\label{sec:q_cssp}
We also present a quantum algorithm for the CSSP under the Hamming distance, based on a simple quantisation of an $\XP$ algorithm due to Marx~\cite{marx_closest_2005}.
We briefly explain the algorithm, and then show how a direct application of Grover search can lead to a quadratic speedup.

The algorithm is based on searching through certain sets of strings called \emph{generators}.
A set $G = \{ g_1,\ldots, g_l\}$ of length-$L$ strings is a \emph{generator} of the length-$L$ string $s$ if, whenever all $g_i$ have the same character at some position $p$, then $s$ also has that character at position $p$.
The \emph{size} of $G$ is $l$ and the \emph{conflict} is the number of positions where the $g_i$ do \emph{not} all agree. 
Marx then proves that if an instance of the CSSP is solvable, then there exists a solution with a generator that is not too large and has low conflict.
We restate the result in~\Cref{lem:generator}.

\begin{lemma}[\cite{marx_closest_2005}] \label{lem:generator}
    If an instance of $CSSP$ is solvable, then there exists a solution $\hat{s}$ that has a generator $G = \{g_1,\ldots,g_l\}$ satisfying:
    \begin{itemize}
        \item All generator strings $g_i$ are length-$L$ substrings of an input string $s_j$.
        \item The size of $G$ is $l \leq \log(d) + 2 $.
        \item The conflict of $G$ is at most $d(\log d + 2)$.
    \end{itemize}
\end{lemma}

The algorithm then proceeds as follows.
First, create a set $T$ of all length-$L$ substrings of $s_1,\,\ldots,\,s_k$.
For every subset $G \subseteq T$ with $\abs{G} = \log(d) + 2$, $G$ is a candidate generator for~\Cref{lem:generator}.
There are at most $\binom{\abs{T}}{\log(d) + 2}$ such subsets.
Since $\abs{T} \leq k \cdot (n - L) \leq kn$ there are at most $\bigo{(kn)^{\log(d) + 2}}$ possible generators.
For each $G$ with conflict at most $d(\log(d) + 2)$, each string $s$ generated by $G$ is a candidate solution.
There are at most $\sigma^{d(\log(d) + 2)}$ such strings.
Checking a candidate string $s$ requires calculating the Hamming distance between $s$ and each length-$L$ substring of each input string $s_i$.
This takes time $\mathcal{O}(k \cdot (n - L) \cdot L) = \mathcal{O}(kn^2)$.
The overall time is thus $\bigo{\sigma^{d(\log(d) + 2)} )(kn)^{\log(d)}
k^3 n^4}$.

We note that the algorithm is simply an exhaustive search over subsets $G \subseteq T$ and then over strings $s$ generated by $G$.
Applying a nested Grover search subroutine to both of these searches leads to a quantum algorithm with runtime 
$\mathcal{\tilde{O}}\Bigl(\sqrt{(kn)^{\log(d) + 2}} $ $\cdot \sqrt{\sigma^{d(\log(d) + 2)}} $ $ \cdot kn^2\Bigr)$. 
We therefore obtain~\Cref{prop:cssp_hamming}.
\cssphamming*

\section{Discussion and Open Questions}
\label{sec:discussion}
We now have 3 quantum algorithms for the CSP. We summarise their properties in~\Cref{tab:q_csp}.
\begin{table}[t]
\begin{tabular}{@{}lrrr@{}}
\toprule
Algorithm & Distance metric & Runtime & Optimality conditions \\ \midrule
Trivial Search & Any weighted edit & $\sigma^{n/2} \cdot \sqrt{k} \cdot D(n)$ & $\sigma$ small, $d$ and $k$ large \\
Random Walk & Hamming & $kn + k^2 d^{(d+5)/2}$ & $d$ small \\
Dynamic Programming & Levenshtein & $\sigma \cdot 2^k \cdot T(n)^{2k}$ & $k$ small \\ \bottomrule
\end{tabular}
\caption{Table summarising the applicability and runtime of the 3 quantum algorithms presented for the CSP. $D(n)$ is the classical time to compute the desired distance metric. $n/e \leq T(n) < n$.}
\label{tab:q_csp}
\end{table}

For concreteness, consider a bioinformatics application where we wish to find a consensus sequence for a number of individuals from the same species. 
In this case, we would expect $d \sim \beta n$ for some constant $\beta \in [0.001,0.1]$, depending on the species.
For example, for humans, there is roughly $0.1\%$ variance between individuals~\cite{auton_global_2015}.
We also take $\sigma = 4$.
We then consider the runtimes as $k$ varies.
For $k = \bigo{1}$, there is a clear advantage for the dynamic programming algorithm, which is polynomial-time, whereas the other algorithms are at least exponential.
Conversely, for $k = \bigo{\text{exp}(n)}$, the random walk algorithm remains only slightly super-exponential, with $\bigo{2^{\beta n\log(n)}}$ scaling, while dynamic programming is now doubly-exponential.
We summarise the runtimes across the range of $k$ in~\Cref{tab:q_csp_runtime}.
The optimal algorithm will depend on the exact set of parameters for the problem instance under consideration.

\begin{table}[t]
\begin{tabular}{@{}lrlrl@{}}
\toprule
\multirow{2}{*}{Algorithm} & \multicolumn{4}{c}{$k$} \\
 & \multicolumn{2}{c}{$\bigo{1}$} & \multicolumn{2}{c}{$\bigo{\log(n)}$} \\ \midrule
Trivial Search & $2^n$ & exp. & $2^n$ & exp. \\
Random Walk & $2^{\beta n\log(n)/2}$ & \begin{tabular}[c]{@{}l@{}}slightly\\ super-exp.\end{tabular} & $2^{\beta n\log(n)/2}$ & \begin{tabular}[c]{@{}l@{}}slightly\\ super-exp.\end{tabular} \\
\begin{tabular}[c]{@{}l@{}}Dynamic \\ Programming\end{tabular} & $T(n)^{\bigo{1}}$ & poly. & $T(n)^{ \bigo{\log(n)} }$ & super-poly. \\ \bottomrule
\end{tabular}
\\[12pt]
\begin{tabular}{@{}lrlrl@{}}
\toprule
\multirow{2}{*}{Algorithm} & \multicolumn{4}{c}{$k$} \\
 & \multicolumn{2}{c}{$\bigo{n^p}$} & \multicolumn{2}{c}{$\bigo{2^{qn}}$} \\ \midrule
Trivial Search & $2^n$ & exp. & $2^{n + qn/2}$ & exp. \\
Random Walk  & $2^{\beta n\log(n)/2}$ & \begin{tabular}[c]{@{}l@{}}slightly\\ super-exp.\end{tabular} & $2^{\beta n\log(n)/2}$ & \begin{tabular}[c]{@{}l@{}}slightly\\ super-exp.\end{tabular} \\
\begin{tabular}[c]{@{}l@{}}Dynamic \\ Programming\end{tabular}  & $T(n)^{2 n^p}$ & super-exp. & $ T(n)^{ 2^{qn+1} }$ & doubly exp. \\ \bottomrule
\end{tabular}
\caption{Tables summarising the runtime of the 3 quantum algorithms presented for the CSP in a bioinformatics context where $d = \beta n$ for a constant $\beta \in [0.001, 0.1]$ and $\sigma = 4$, split into poly-scaling and exponential-scaling regimes.. Only the dominant factors are shown -- log factors are hidden for poly-scaling algorithms, and poly factors are hidden for super-poly-scaling algorithms. We use slightly super-exponential to describe an ${2^{\bigo{n\log(n)}}}$ scaling, following~\cite{lokshtanov_slightly_2018}. $n/e \leq T(n) < n$.}
\label{tab:q_csp_runtime}
\end{table}

We also have an algorithm for the CSSP. It is best suited to bounded-alphabet settings where the threshold distance is not too large.
For example, the bioinformatics setting of $d \ll n$ and $\sigma = 4$ results in a tractable runtime up to moderately large $k$ and $n$.

Similarly to the classical case, our lower bound for the CSP holds only for binary alphabets. 
Extending the result to general alphabets would most likely require finding a matching classical reduction in this setting, where the equivalence to the Remotest String Problem no longer holds.

A further interesting open question is to find lower bounds for the various parameterized settings. 
For example, is the near-quadratic improvement optimal in the bounded-$d$ setting?
Finding such bounds would likely require using the full QSETH framework provided by Buhrman et al.~\cite{buhrman_framework_2021}.
It has been shown classically that reductions from branching programs (equivalently, $\NC$ circuits) can lead to stronger fine-grained complexity results~\cite{abboud_simulating_2015}.
For example, a subquadratic edit distance algorithm would imply exponential improvements over exhaustive search for certain complex SAT problems. 
Similarly, a reduction leads to an $\Omega(n^{1.5})$ quantum lower bound for computing the edit distance, conditioned on $\NC$-QSETH, whereas the basic QSETH gives only a $\Omega(n)$ lower bound.
Finding the right class of circuits on which to condition QSETH to find parameterized lower bounds remains an intriguing problem.

\section*{Acknowledgements}
JC and SS thank the entire QPG team for their useful support and discussions throughout the process of this work, and especially Hitham Hassan for feedback on draft versions. Work on the Quantum Pangenomics project was supported by Wellcome Leap as part of the Q4Bio Program. SS was also supported by the Royal Society University Research Fellowship.

\bibliographystyle{ieeetr}
\bibliography{references.bib}

\appendix
\section{Quantum Search via Random Walk}
\label{app:quantum_search}
In this Appendix, we give a full background on the quantum search via random walk method due to Magniez et al.~\cite{magniez_search_2011}, including a brief introduction to the theory of Markov chains.
For a full treatment, see, e.g.,~\cite{mcnew_eigenvalue_2011}.

\subsection{Markov Chains, Eigenvalue Gaps and Congestion}

A Markov chain is a stochastic process in which the probability of each event depends only on the state attained after the previous event.
The chain has a state space $\Omega$ consisting of all the states the chain can be in.
We consider discrete time chains, which can be described by a sequence of random variables $X_0, X_1, X_2,\ldots$ taking values in $\Omega$.

We identify a chain with its transition matrix $P = (p_{xy})_{x,y \in \Omega}$, where $p_{xy}$ is the probability of a transition from $x$ to $y$.
We must have $p_{xy} \geq 0$ and $\sum_{y \in \Omega} p_{xy} = 1$ to have the correct probabilistic interpretation.
The eigenvalues of $P$ therefore have magnitude at most 1.
It is often useful to associate a Markov chain with a random walk on an edge-weighted, directed graph $G = (\Omega, E)$ with weight on edge $x \rightarrow y$ given by $p_{xy}$.

The \emph{mixing time} of a Markov chain is the minimum time $t_{\text{mix}}$ such that the time-$t_{\text{mix}}$ distribution is ``close'' to the stationary distribution $\pi$.
The mixing time plays a central role in the time complexity of searching by random walks.
For example, consider~\Cref{code:c_search} with $t_1 \approx t_{\text{mix}}$.
Intuitively, leaving the chain for $t_{\text{mix}}$ time steps results in a well-mixed distribution, close to $\pi$.
Repeating for $t_2$ steps represents repeated sampling from $\pi$.
If we have a good bound on the probability under $\pi$ of sampling a marked element, we can use a Chernoff bound to argue that~\Cref{code:c_search} succeeds with high probability.
In order to bound the mixing time of chains, we will need several definitions.

A chain is said to be \emph{irreducible} if every state is reachable from every other state.
A state $x$ has \emph{period} given by $\gcd\{ n > 0 : \mathbb{P}(X_n = x \ | \ X_0 = x) > 0 \}$. 
A state is \emph{aperiodic} if it has period $1$, and an irreducible chain is \emph{aperiodic} if any state is.
An irreducible, aperiodic chain is \emph{ergodic}.
A state $x$ is \emph{positive recurrent} if, starting from $x$ at time $0$, the chain will return to $x$ in finite time with probability $1$.
An irreducible chain is \emph{positive recurrent} if any state is.

A \emph{stationary distribution} $\pi = (\pi_x)$ of a Markov chain $P$ is a left eigenvector of $P$ with eigenvalue $1$ and with entries satisfying $\pi_x \geq 0$ and $\sum_x \pi_x = 1$.
Ergodic chains always have a unique stationary distribution.
Moreover, all other eigenvalues have magnitude strictly less than 1.

For an ergodic chain $P$, the time-reversed Markov chain $P^* = (p^*_{xy})$ of $P$ is defined by the equations $\pi_x p_{xy} = \pi_y p^*_{yx}$. 
$P$ is said to be \emph{reversible} if $P^* = P$.
A reversible chain has real eigenvalues.
We can check reversibility by using~\Cref{thm:kolmogorov}.
\begin{theorem}[Kolmogorov's Criterion~\cite{kelly_reversibility_2011}]\label{thm:kolmogorov}
    An ergodic, positive recurrent Markov Chain $P$ with finite state space $\Omega$ is reversible if and only if $P$ satisfies
    \begin{equation}\label{eq:kolmogorov}
        p_{x_1 x_2} p_{x_2 x_3} \ldots p_{x_{n-1}x_{n}} p_{x_n x_1} = 
    p_{x_1 x_n} p_{x_n x_{n-1}} \ldots p_{x_{3}x_{2}} p_{x_2 x_1}
    \end{equation}
    for all finite sequences of states $x_1,\,\ldots,\,x_n \in \Omega$.\\
    In words, \Cref{eq:kolmogorov} states that for any closed loop path, the probability of traversing the loop in the forwards and backwards directions must be equal.
\end{theorem}
A Markov chain $P$ is \emph{lazy} if the probability of remaining in any state is at least $1/2$. A lazy, reversible Markov chain has non-negative eigenvalues.
We can therefore write the eigenvalues of $P$ in descending order as $\lambda_1 = 1 > \lambda_2 \geq \ldots \lambda_n \geq 0$.

We can then define the \emph{eigenvalue gap} of the chain by $\delta(P) \coloneqq \abs{\lambda_1} - \abs{\lambda_2} = 1 - \lambda_2$.
The eigenvalue gap is closely related to the mixing time of a Markov chain. 
Indeed, we have lower and upper bounds for the mixing time:
\begin{theorem}[e.g.,~\cite{levin_markov_2006}]
    Let $P$ be a reversible Markov chain with stationary distribution $\pi$.
    Let $\pi_{\min} \coloneqq \min_{x \in \Omega} \pi_x$.
    Then,
    \begin{equation}
    \log(2) \left(\frac{1}{\delta(P)} - 1\right)
    \leq t_{\text{mix}} \leq
    \frac{1}{\delta(P)} \log(\frac{4}{\pi_{\min}}   ) .
    \end{equation}
\end{theorem}
While this is a tight bound, the eigenvalue gap can be prohibitively hard to compute in practice.
A common tool for bounding the eigenvalue gap is a property known as the \emph{conductance}.
Loosely, the conductance is a measure of how well connected the chain is.

We first define the \emph{ergodic flow} $Q(S,T)$ between subsets $S, T \subset \Omega$ by
$Q(S,\,T) \coloneqq \sum_{x \in S,\, y \in T} \pi_x p_{xy}$.
For a subset $S$, we also define
$\pi(S) \coloneqq \sum_{x \in S} \pi_x$ 
and the \emph{bottleneck ratio}
$\Phi(S) \coloneqq \frac{Q(S,\,S^c)}{\pi(S)}$.
The \emph{conductance} $\Phi$ is then given by:
\begin{equation}
    \Phi \coloneqq \min_{\substack{S \subset \Omega \\ \pi(S) \leq 1/2}} \Phi(S).
\end{equation}

A result of Alon and Milman~\cite{alon_1_1985} bounds the eigenvalue gap by the conductance.
\begin{theorem}[\cite{alon_1_1985}]\label{thm:conductance}
    For a lazy, reversible Markov chain $P$, the eigenvalue gap $\delta(P)$ satisfies:
    \begin{equation}
         \frac{\Phi^2}{2} \leq \delta(P) \leq 2\Phi.
    \end{equation}
\end{theorem}
The conductance remains hard to compute in many cases. We therefore require one further bound, via the \emph{congestion} $\rho$.

For any pair of states $x,\,y$, define a canonical path $\gamma_{xy} = (x = z_0,\,z_1,\,\ldots,\,z_{k-1},\,z_k = y)$.
The canonical path need not be the shortest path between the 2 points, but that will suffice for our purposes.

Let $\Gamma = \{\gamma_{xy} : x,\,y \in X \}$ be the set of all canonical paths.
The congestion is given by:
\begin{equation}
\label{eq:app_congestion}
    \rho = \rho(\Gamma) = \max_{(u,\,v) \in E} \left\{
\frac{1}{\pi_u p_{uv}} \sum_{\substack{x,\,y \in \Omega: \\ u,\,v \in \gamma_{xy}}} \pi_x \pi_y
\right\}.
\end{equation}
The $\pi_u p_{uv}$ term is intuitively the ``capacity'' of the edge $(u,\,v)$ in the graph, in that it relates how much traffic the edge sees in the stationary distribution.
The sum then counts the ``load'' the edge experiences within the family of canonical paths.
The congestion is then the maximum load of any edge, normalised by its capacity.
As expected, high congestion corresponds to low conductance:
\begin{theorem}[e.g.,~\cite{mcnew_eigenvalue_2011}]
    \label{thm:congestion}
    For any reversible Markov chain and any choice of canonical paths $\Gamma$:
    \begin{equation}
        \Phi \geq \frac{1}{2 \rho(\Gamma)}.
    \end{equation}
\end{theorem}

Combining~\Cref{thm:conductance}~\&~\Cref{thm:congestion}, we may trivially obtain:

\begin{theorem}\label{thm:eigenvalue_gap}
    For a lazy, reversible Markov chain $P$ and any choice of canonical paths $\Gamma$, the eigenvalue gap $\delta(P)$ is bounded by:
    \begin{equation}
        \delta(P) \geq \frac{1}{8 \rho(\Gamma)^2}.
    \end{equation}
\end{theorem}

\subsection{A Quantum Search Algorithm}
\label{app:q_search_algo}
In this~\namecref{app:q_search_algo}, we briefly summarise the functioning of the quantum search algorithm due to Magniez et al~\cite{magniez_search_2011}.

Having prepared a state $\ket{\pi}$ that encodes the stationary distribution (with cost $S$), the goal is to prepare the normalised projection of $\ket{\pi}$ onto the marked subspace, which we call $\ket{\mu}$.
The standard, Grover-esque method is to implement a rotation in the 2D space $\mathcal{S} \coloneqq \text{span}(\ket{\pi}, \ket{\mu})$.
As usual, this rotation is effected by composing 2 reflections.

The first reflection is about a state $\ket{\mu^\perp}$, a state orthogonal to $\ket{\mu}$. 
In $\mathcal{S}$, this is equivalent (up to a phase) to reflecting about the marked subspace.
This can be done at the cost of checking whether an element is marked, accruing a complexity of $C$.

The second reflection is about the state $\ket{\pi}$.
This could be done by rotating $\ket{\pi}$ to a fixed reference state (e.g. $\ket{0}$), reflecting about that state, and then rotating back.
This induces a cost of $2(S+U)$, since it requires the operator used to prepare $\ket{\pi}$ and its inverse.
In general, $S$ is costly and should be avoided.
Instead, the authors propose to use a quantum walk operator that can be implemented in a cost of only $4U$ combined with a phase estimation procedure.
Since the chain is known to have a unique eigenvalue $1$ corresponding to the eigenvector $\pi$, the phase estimation algorithm can distinguish between parts of the state aligned with or orthogonal to $\ket{\pi}$. 

By deriving a quadratic relation between the eigenvalue gap of the chain $P$ and the gap in the phase of the quantum walk, it can be shown that only $\bigo{1/\sqrt{\delta}}$ calls to the walk are needed.
The $\bigo{1/\sqrt{\epsilon}}$ scaling follows from a technical requirement to use a randomized Grover's algorithm~\cite{boyer_tight_1998}, since the probability of a marked element is not known \emph{a priori}.

\section{A Markov Chain on a Perfect Tree}
\label{app:markov}
We first restate the definition of $P$ from the main text.

Let $\mathcal{T}$ be a perfect $(d+1)$-ary tree with depth $d$.
Let $\Omega$ be the set of nodes of $\mathcal{T}$, given by:
\begin{equation}
    \Omega = \{ \mathbf{j} = (j_1,\ldots,j_l, 0,\ldots, 0) \in \{0,\ldots,d+1\}^{d} : 0 \leq l \leq d; 1 \leq j_i \leq d \}.
\end{equation}

The transition probabilities of $P$ are:
\begin{align}
\text{If $\mathbf{j}$ is neither root nor leaf,} \quad
    p_{\mathbf{jj'}} &= \begin{cases}
    \frac{1}{2} & \text{if } \mathbf{j'} = \mathbf{j}, \\
    \frac{1}{4(d+1)} & \text{if $\mathbf{j'}$ is a child of $\mathbf{j}$}, \\  
    \frac{1}{4} & \text{if $\mathbf{j'}$ is the parent of $\mathbf{j}$}, \\
    0 & \text{otherwise.}
\end{cases} \\
\text{If $\mathbf{j}$ is a leaf,} \quad
p_{\mathbf{jj'}} &= \begin{cases}
    \frac{1}{2} & \text{if } \mathbf{j'} = \mathbf{j}, \\
    \frac{1}{2} & \text{if $\mathbf{j'}$ is the parent of $\mathbf{j}$}, \\
    0 & \text{otherwise.}
\end{cases} \\
\text{If $\mathbf{j}$ is the root,} \quad
p_{\mathbf{jj'}} &= \begin{cases}
    \frac{1}{2} & \text{if } \mathbf{j'} = \mathbf{j}, \\
    \frac{1}{2(d+1)} & \text{if $\mathbf{j'}$ is a child of $\mathbf{j}$}, \\
    0 & \text{otherwise.}
\end{cases}
\end{align}
The set of edges of $P$ is $E = \{(\mathbf{j}, \mathbf{j'}) \in \Omega \times \Omega: j_i = j'_i, i = 1,\ldots, l; j_{l+1} = 0; l = d-1 \lor j'_{l+2} = 0 \}$.

$P$ is easily seen to be lazy and ergodic.
Using Kolomogorov's Criterion~(\Cref{thm:kolmogorov}), $P$ is also reversible due to the following.
All closed loops must consist of a path that doubles back on itself, since the in-degree of each vertex is 1.
Therefore, for any transition probability $p_{\mathbf{jj'}}$ that appears in the forward direction, there must also be a corresponding $p_{\mathbf{j'j}}$.
Since the backward direction is formed by swapping the direction of each transition in the forward direction, simply interchanging each of the above pairs of terms results in the backward direction matching the forward, as required. 

To compute the stationary distribution, we use the symmetry of the walk to infer that for any nodes $\mathbf{j},\mathbf{j'}$ in the same layer, $\pi_{\mathbf{j}} = \pi_\mathbf{j'}$.
It is therefore useful to consider a simpler Markov chain $Q$ defined on the layers of the tree, with state space $\chi = \{l : 0 \leq l \leq d\}$.
$Q$ is coupled to $P$ such that if $P$ is in state $\mathbf{j}$ in level $l$, then $Q$ is in state $l$.
If the stationary distribution of $Q$ is $\tilde{\pi}$, then $\pi$ can be recovered by distributing the probability mass $\tilde{\pi}_l$ equally among all the nodes at layer $l$.

The transition probabilities of $Q$ are given by the tri-diagonal matrix:
\begin{equation}
    Q = \begin{bmatrix}
    \frac{1}{2} & \frac{1}{2} & \\
    \frac{1}{4} & \frac{1}{2} & \frac{1}{4} & \\
    & \frac{1}{4} & \frac{1}{2} & \frac{1}{4} & \\
    & & & \ddots & \\
    & & & & \frac{1}{4} & \frac{1}{2} & \frac{1}{4} & \\
    & & & & & \frac{1}{2} & \frac{1}{2}
\end{bmatrix}.
\end{equation}
We may further simplify by writing $Q = \frac{1}{2}(I + \tilde{Q})$, and noting that $\tilde{\pi} Q = \tilde{\pi} \Leftrightarrow \tilde{\pi} \tilde{Q} = \tilde{\pi}$.
$\tilde{Q}$ defines a symmetric random walk with reflecting barriers.
A direct computation yields:
\begin{equation}
\label{eq:stationary_distribution}
    \tilde{\pi}_l = \begin{cases}
        \frac{1}{2d} & \text{if } l = 0, \ d, \\
        \frac{1}{d} & \text{if } l = 1,\ldots, d-1.
    \end{cases}
\end{equation}
For any pair of nodes $\mathbf{j}, \mathbf{j'}$, take the canonical bath $\gamma_{\mathbf{jj'}}$ to be the shortest path from $\mathbf{j}$ to $\mathbf{j'}$.
We may now compute the congestion of $P$, which is given by~\Cref{eq:congestion}:
\begin{equation}
\label{eq:congestion}
    \rho = \max_{(\mathbf{j}, \mathbf{j}') \in E} \left\{
    \rho_{\mathbf{jj'}} \coloneqq
\frac{1}{\pi_\mathbf{j} p_{\mathbf{jj'}}} \sum_{\substack{\mathbf{k},\mathbf{k'} \in \Omega: \\ \mathbf{j},\,\mathbf{j'} \in \gamma_{\mathbf{kk'}}}} \pi_\mathbf{k} \pi_\mathbf{k'}
\right\}.
\end{equation}
For a self-edge $(\mathbf{j}, \mathbf{j})$, the congestion $\rho_{\mathbf{jj}}$ is 0 since no canonical path goes through any self-edges.

For an edge between distinct nodes $\mathbf{j}, \mathbf{j'}$, assume $\mathbf{j}$ is in layer $l$ and $\mathbf{j'}$ is in layer $l+1$.
Then $\mathbf{j},\mathbf{j'} \in \gamma_{\mathbf{kk'}}$ if and only if $\mathbf{k}$ equals or is a descendent of $\mathbf{j'}$ and $\mathbf{k'}$ is not a descendant of $\mathbf{j'}$, or vice versa.

For $\mathbf{j'},\,\mathbf{k} \in \Omega$, write $\mathbf{k} \preceq \mathbf{j'}$ if $\mathbf{k} = \mathbf{j'}$ or $\mathbf{k}$ is a descendant of $\mathbf{j'}$.
Then, for $\mathbf{j}$ not the root, we have:
\begin{align}
    \rho_{\mathbf{jj'}} &= 
    \frac{1}{\pi_\mathbf{j} p_{\mathbf{jj'}}} 
    \left(
    \sum_{\substack{
    \mathbf{k}: \mathbf{k} \preceq \mathbf{j'}\\
    \mathbf{k'}: \mathbf{k'} \npreceq \mathbf{j'}
    }} 
    \pi_\mathbf{k} \pi_\mathbf{k'}
    +
    \sum_{\substack{
    \mathbf{k'}: \mathbf{k'} \preceq \mathbf{j'}\\
    \mathbf{k}: \mathbf{k} \npreceq \mathbf{j'}
    }} 
    \pi_\mathbf{k} \pi_\mathbf{k'}
    \right)\\
    &= 
    d(d+1)^l \cdot 4(d+1) \cdot 2 
    \sum_{\mathbf{k}: \mathbf{k} \preceq \mathbf{j'}}\pi_\mathbf{k}
    \sum_{\mathbf{k'}: \mathbf{k'} \npreceq \mathbf{j'}}\pi_\mathbf{k'}.
\end{align}
For $\mathbf{j} = \mathbf{0}$, the root, we have a slightly modified expression, since $\pi_\mathbf{0}$ is reduced by a factor of 2 and $p_{\mathbf{0j'}}$ is increased by a factor of 2:
\begin{align}
    \rho_{\mathbf{0j'}} &= 
    2d(d+1)^l \cdot 2(d+1) \cdot 2 
    \sum_{\mathbf{k}: \mathbf{k} \preceq \mathbf{j'}}\pi_\mathbf{k}
    \sum_{\mathbf{k'}: \mathbf{k'} \npreceq \mathbf{j'}}\pi_\mathbf{k'}.
\end{align}
Clearly these modifications cancel, and we may treat the expressions on the same footing.
Then we may compute the sums:
\begin{align}
    \sum_{\mathbf{k}: \mathbf{k} \preceq \mathbf{j'}}\pi_\mathbf{k}
    &=
    \frac{\tilde{\pi}_{l+1}}{(d+1)^{l+1}}
    + \frac{(d + 1)\tilde{\pi}_{l+2}}{(d+1)^{l+2}}
    + \frac{(d + 1)^2\tilde{\pi}_{l+3}}{(d+1)^{l+3}} + \ldots +
    \frac{(d + 1)^{d-l-1}\tilde{\pi}_{d}}{(d+1)^{d}}\\
    &=
    \frac{1}{d(d+1)^{l+1}}\left(
        (d-l-1) + \frac{1}{2}
    \right)\\
    &= \frac{2(d-l)-1}{2d(d+1)^{l+1}}.\\
    \sum_{\mathbf{k'}: \mathbf{k'} \npreceq \mathbf{j'}}\pi_\mathbf{k'}
    &=
    1 - \sum_{\mathbf{k}: \mathbf{k} \preceq \mathbf{j'}}\pi_\mathbf{k}\\
    &=
    \frac{2d(d+1)^{l+1} - 2(d-l) +1}{2d(d+1)^{l+1}}.
\end{align}
Noting that these expressions depend only $l$ and $d$, we may write $\rho_l$ for the congestion of an edge from layer $l$ to layer $l+1$.
Finally, we have:
\begin{align}
    \rho_l &= 8d(d+1)^{l+1}\cdot \frac{
    \bigl(2(d-l)-1\bigr) \bigl(2d(d+1)^{l+1} - 2(d-l) +1 \bigr)
    }{4d^2(d+1)^{2l+2}}\\
    &=
    \frac{2\bigl(2(d-l)-1\bigr) \bigl(2d(d+1)^{l+1} - 2(d-l) +1 \bigr)}{d(d+1)^{l+1}}.
\end{align}
This is maximised for $l = 0$, where the congestion takes the value
\begin{align}
    \rho_0 = \frac{2(2d-1)(2d^2+1)}{d(d+1)} < 8d.
\end{align}
By~\Cref{thm:eigenvalue_gap}, the eigenvalue gap $\delta(P)$ of $P$ is bounded by
\begin{equation}
    \delta(P) \geq \frac{1}{8\rho^2} > \frac{1}{2^9 d^2}.
\end{equation}
\end{document}